\documentclass{llncs}
\usepackage{amsmath,epsfig}
\usepackage{fullpage}
\usepackage{listings}
\usepackage{caption}
\usepackage{subcaption}

\usepackage{graphicx}

\usepackage[ruled,vlined,linesnumbered]{algorithm2e}
\usepackage{qtree}
\usepackage{subfig}
\usepackage{url}

\usepackage{wrapfig,lipsum,booktabs}

\usepackage[usenames,dvipsnames]{pstricks}
\usepackage{epsfig}
\usepackage{pst-grad} 
\usepackage{pst-plot} 

\newcommand{\qedsymb}{\hfill{\rule{2mm}{2mm}}}

\newcommand{\remove}[1]{}

\DeclareMathOperator{\sa}{\mathit{SA}}
\DeclareMathOperator{\rank}{\mathit{Rank}}
\DeclareMathOperator{\lcp}{\mathit{LCP}}

\DeclareMathOperator{\lr}{\mathit{LR}}
\DeclareMathOperator{\llr}{\mathit{LLR}}
\DeclareMathOperator{\llrc}{\mathit{LLRc}}
\DeclareMathOperator{\llrs}{\mathit{LLRS}}
\DeclareMathOperator{\llrr}{\mathit{LLRr}}
\DeclareMathOperator{\lrs}{\mathit{LRS}}
\DeclareMathOperator{\ct}{\mathit{Count}}
\DeclareMathOperator{\ps}{\mathit{Prefix\_Sum}}

\title{On Longest Repeat Queries Using GPU\thanks{Authors names are
    listed in alphabetical order.}}

\author{Yun Tian  
\and Bojian Xu
\thanks{Corresponding
    author. Phone: +1 (509)359-2817. Fax: +1 (509)359-2215.
    Supported in part by EWU Faculty Grants for
    Research and Creative Works. 
}
}

\institute{
Department of Computer Science\\Eastern Washington University, Cheney,
WA 99004, USA\\
\email{ytian@ewu.edu, bojianxu@ewu.edu}
}

\begin{document}

\maketitle

\begin{abstract}
  Repeat finding in strings has important applications in subfields
  such as computational biology. 
  The challenge of finding the longest repeats covering
  particular string positions was recently proposed and solved by \.{I}leri et
  al., using a total of the optimal $O(n)$ time and space, where $n$
  is the string size.  However, their solution can only find the
  \emph{leftmost} longest repeat for each of the $n$ string position.
  It is also not known how to parallelize their solution.
  In this paper, we propose a new solution for longest repeat
  finding, which although is theoretically suboptimal in time but is
  conceptually simpler and works faster and uses less memory space in
  practice than the optimal solution.  Further, our solution can find
  \emph{all} longest repeats of every string position, while still
  maintaining a faster processing speed and less memory space
  usage. Moreover, our solution is \emph{parallelizable} in the shared
  memory architecture (SMA), enabling it to take advantage of the
  modern multi-processor computing platforms such as the
  general-purpose graphics processing units (GPU).  We have
  implemented both the sequential and parallel versions of our
  solution.  Experiments with both biological and non-biological data
  show that our sequential and parallel solutions are faster than the
  optimal solution by a factor of 2--3.5 and 6--14, respectively, and
  use less memory space.
\end{abstract}


\section{Introduction}
\label{sec:intro}
Repetitive structures and regularities finding in genomes and proteins
is important as these structures play important roles in the
biological functions of genomes and
proteins~\cite{Gus97,Smyth2013-survey,saha08.2,saha08.1,LW06,Benson99,BDH09,KurtzSch99}.
It is well known that overall about one-third of the whole human
genome consists of repeated subsequences~\cite{McC1993}; about
10--25\% of all known proteins have some form of repetitive
structures~\cite{LW06}. In addition, a number of significant problems
in molecular sequence analysis can be reduced to repeat
finding~\cite{Mar83}. Another motivation for finding repeats is to
compress the DNA sequences, which is known as one of the most
challenging tasks in the data compression field. DNA sequences consist
only of symbols from {\tt \{ACGT\}} and therefore can be represented
by two bits per character. Standard compressors such as {\tt gzip} and
{\tt bzip} usually use more than two bits per character and therefore
cannot achieve good compression. Many modern genomic sequence data
compression techniques highly rely on the repeat finding in
sequences~\cite{MR04,BF05}.

The notion of maximal
repeat and super maximal repeat~\cite{Gus97,BDH09,KVX-tcbb2012,BBO-spire2012} captures all the
repeats of the whole string in a space-efficient manner, but it does
not track the locality of each repeat and thus can not support the
finding of repeats that cover a particular string position.  
For this reason, \.{I}leri et al.~\cite{IKX-repeat-CORR2015} proposed the
challenge of longest repeat query, which is to find the longest
repetitive substring(s) that covers a particular string position.
Because any substring of a repetitive substring is also repetitive,
the solution to longest repeat query effectively provides an effective
``stabbing'' tool for finding the majority of the repeats covering a
string position.  \.{I}leri et al.\ proposed an $O(n)$ time and space
algorithm that can find the \emph{leftmost} longest repeat of every
string position. Since one has to spend $\Omega(n)$ time and space to
read and store the input string, the solution of \.{I}leri et al.\ is
optimal.

\paragraph{Our contribution.}
In this paper, we propose a new solution for longest repeat query.
Although our solution is theoretically suboptimal in the time cost,
it is conceptually simpler and runs faster and uses less memory
space than the optimal solution in practice.  Our solution can also
find \emph{all} longest repeats for every string position while
still maintaining a faster processing speed and less space usage, whereas
the optimal solution can only find the leftmost candidate.  Further,
our solution can be parallelized in the shared-memory architecture,
enabling it to take advantage of the modern multi-processor computing
platforms such as the general-purpose graphics processing units
(GPU)~\cite{Che2008,Nickolls2010}.  We have implemented both the
sequential and parallel versions of our solution. Experiments with
both biological and non-biological data show that our solution run
faster than the $O(n)$ optimal solution by a factor of 2--3.5 using
CPU and 6--14 using GPU, and use less space in both settings.

\paragraph{Road map.}
After formulating the problem of longest query in
Section~\ref{sec:prob}, we prepare some technical background and
observations in Section~\ref{sec:prelim} for our solutions.
Section~\ref{sec:seq} presents the sequential version of our
solutions. Following the interpretation in Section~\ref{sec:seq}, it
is natural and easy to get the parallel version of our solution, which
is presented in Section~\ref{sec:para}. Section~\ref{sec:exp} shows
the experimental results on the comparison between our solutions and the $O(n)$
solution using real-world data.

\section{Problem Formulation}
\label{sec:prob}

We consider a {\bf string} $S[1\ldots n]$,
 where each character $S[i]$ is
drawn from an alphabet $\Sigma=\{1,2,\ldots, \sigma\}$.
A {\bf substring} $S[i\ldots j]$
of $S$ represents $S[i]S[i+1]\ldots S[j]$ if $1\leq i\leq j \leq n$,
and is an empty string if $i>j$.
String $S[i'\ldots j']$ is a {\bf proper substring} of another string
$S[i\ldots j]$ if $i\leq i' \leq j' \leq j$ and $j'-i' < j-i$. 
%
%
The {\bf length} of a non-empty substring $S[i\ldots j]$, denoted as
$|S[i\ldots j]|$, is $j-i+1$. We define the length of an empty string
as zero. 
A {\bf prefix} of $S$ is a substring $S[1\ldots i]$
for some $i$, $1\leq i\leq n$. 
A {\bf proper prefix} $S[1\ldots i]$ is a prefix of $S$ where $i <
n$.
A {\bf suffix} of $S$ is a substring
$S[i\ldots n]$ for some $i$, $1\leq i\leq n$.  
A {\bf proper suffix} $S[i\ldots n]$ is a suffix of $S$ where $i >
1$.
We say the character $S[i]$ occupies the string {\bf position} $i$.
We say the substring $S[i\ldots j]$ {\bf covers} the $k$th position of
$S$, if $i\leq k \leq j$.  
For two strings $A$ and $B$, we write ${\bf A=B}$ (and say $A$ is {\bf
  equal} to $B$), if $|A|= |B|$ and $A[i]=B[i]$ for 
$i=1,2,\ldots, |A|$.  
%
We say $A$ is lexicographically smaller than $B$,
denoted as ${\bf A < B}$, if (1) $A$ is a proper prefix of $B$, or (2)
$A[1] < B[1]$, or (3) there exists an integer $k > 1$ such that
$A[i]=B[i]$ for all $1\leq i \leq k-1$ but $A[k] < B[k]$.
A substring
$S[i\ldots j]$ of $S$ is {\bf unique}, if there does not exist
another substring $S[i'\ldots j']$ of $S$, such that 
$S[i\ldots j] = S[i'\ldots j']$ but $i\neq i'$. 
A substring is a {\bf repeat} if it is not unique.
A character $S[i]$ is a {\bf singleton}, if it appears only once in
$S$.

\begin{definition}
\label{def:lr}
For a particular string position $k\in \{1,2,\ldots, n\}$,  
the {\bf longest repeat (LR) covering position} ${\bf k}$, denoted
as $\lr_k$,
is 
a repeat substring $S[i\ldots j]$, such that: (1) $i\leq k \leq j$, and 
(2) there does not exist  another repeat substring $S[i'\ldots j']$, such
that $i'\leq k \leq j'$ and $j'-i' > j-i$. 
\end{definition}

Obviously, for any string position $k$, if $S[k]$ is not a singleton,
$\lr_k$ must exist, because at least $S[k]$ itself is a repeat.
Further, there might be multiple choices for $\lr_k$. For example, if
$S={\tt abcabcddbca}$, then $\lr_2$ can be either $S[1\ldots 3]={\tt
  abc}$ or $S[2\ldots 4]={\tt bca}$. 
In this paper, we study the problem of finding the longest repeats of
every string position of $S$.

\bigskip

\noindent{\bf Problem} (longest repeat query): 
For every string position $k\in \{1,2,\ldots,n\}$, 
we want to find $\lr_k$ or the fact that it does not
exist. If multiple choices for $\lr_k$ exist, we want to find all
of them.   

\section{Preliminary}
\label{sec:prelim}

The {\bf suffix array} $\sa[1\ldots n]$ of the string $S$ is a
permutation of $\{1,2,\ldots, n\}$, such that for any $i$ and $j$,
$1\leq i < j \leq n$, we have $S[\sa[i]\ldots n] < S[\sa[j]\ldots n]$.
That is, $\sa[i]$ is the starting position of the $i$th suffix in
the sorted order of all the suffixes of $S$.
The {\bf rank array} $\rank[1\ldots n]$ is the inverse of the suffix
array. That is, $\rank[i]=j$ iff $\sa[j]=i$. 
The {\bf longest common prefix (lcp) array} $\lcp[1\ldots n+1]$ is an
array of $n+1$ integers, such that for $i=2,3,\ldots, n$, $\lcp[i]$ is
the length of the lcp of the two suffixes $S[\sa[i-1]\ldots n]$ and
$S[\sa[i]\ldots n]$. We set $\lcp[1]=\lcp[n+1]=0$.  In the literature,
the lcp array is often defined as an array of $n$ integers. We include
an extra zero at $\lcp[n+1]$ is only to simplify the description 
of our upcoming
algorithms.  
Table~\ref{tab:suflcp} shows the suffix array and the lcp
array of the example string {\tt mississippi}.

\begin{definition}
\label{def:llr}
For a particular string position $k\in \{1,2,\ldots, n\}$, the
{\bf left-bounded longest repeat (LLR) starting at position $k$},
denoted as $\llr_k$,
is a repeat $S[k\ldots j]$,
such that either $j=n$ or $S[k\ldots j+1]$ is unique. 
\end{definition}

Clearly, for any string position $k$, if $S[k]$ is not a singleton,
$\llr_k$ must exist, because at least $S[k]$ itself is a repeat.
Further, if $\llr_k$ does exist, there must be only  one choice,
because $k$ is a fixed string position and the length of $\llr_k$ must
be as long as possible. Lemma~\ref{lem:llr} shows that, given the rank
and lcp arrays of the string $S$, we can directly
calculate any $\llr_k$ or find the fact of its nonexistence.

\begin{table}[t]
\parbox{.40\linewidth}{
\centering
\def\0{\phantom{0}}
{\footnotesize
\begin{tabular}{c|c|c|l}
\toprule
$i$ & $\lcp[i]$  & $\mathit{\sa}[i]$ & suffixes\\
\hline
\hline
\01 & 0 & 11\0  &{\tt i}\\ 
\02 & 1 & \08\0  & {\tt  ippi}\\ 
\03 & 1 & \05\0  & {\tt  issippi}\\ 
\04 & 4 & \02\0  & {\tt  ississippi}\\ 
\05 & 0 & \01\0  & {\tt  mississippi}\\
\06 & 0 & 10\0  & {\tt  pi}\\ 
\07 & 1 &  \09\0  & {\tt ppi}\\
\08 & 0 & \07\0  & {\tt sippi}\\
\09 & 2  & \04\0  & {\tt sissippi}\\
10 & 1  & \06\0  & {\tt ssippi}\\ 
11 & 3 & \03\0  & {\tt ssissippi}\\
12 & 0 & -- & --\\ \bottomrule
\hline
\end{tabular}
}
\caption{The suffix array and the lcp array of an example string $S={\tt mississippi}$.}
\label{tab:suflcp}
}
\hfill
\parbox{.55\linewidth}{
\centering
\def\0{\phantom{0}}
\vspace*{2.7cm}
  \begin{tabular}{|l|l|c|c|c|}
\hline
LLR array type& \#walk steps  & DNA & English  & Protein \\
  \hline
&Minimum  & \0\0\0\0\,1  & \0\0\0\0\0\,1 & \0\0\0\0\,1\\

raw&Maximum & 14,836  & 109,394& 25,822\\

&Average  ($=\alpha$) & \0\0\0\,48  & \0\04,429& \0\0\,215\\

\hline

&Minimum  & \0\0\0\0\,1  & \0\0\0\0\0\,1  & \0\0\0\0\,1  \\

compact &Maximum   & \0\0\0\,14  & \0\0\0\0\,10  & \0\0\0\,35  \\

&Average  ($=\beta$)  &\0\0\0\0\,6  & \0\0\0\0\0\,2  & \0\0\0\0\,3  \\
\hline
   \end{tabular}
  \caption[test]{The number of walk steps in the LR's calculation using our solutions for
    several $50$MB files from Pizza\&Chili~\footnotemark}
  \label{tab:steps}
}
\end{table}

\footnotetext{\url{http://pizzachili.dcc.uchile.cl/texts.html}}

\begin{lemma}
\label{lem:llr}
For $i=1,2,\ldots,n$: 
$$
\llr_i = 
\left \{
\begin{array}{lll}
S[i\ldots i + L_i-1] & , & \textrm{\ \ \ if \ \ } L_i > 0\\
\textit{does not exist} & , & \textrm{\ \ \ if \ \ }L_i = 0
\end{array}
\right.
$$
where $L_i = \max\{\lcp[\rank[i]],\lcp[\rank[i]+1]\}$.
\end{lemma}

\begin{proof}
  Note that $L_i$ is the length of the lcp between the suffix
  $S[i\ldots n]$ and any other suffix of $S$.  If $L_i > 0$, it means
  substring $S[i\ldots i+L_i-1]$ is the lcp among $S[i\ldots n]$ and any
  other suffix of $S$. So $S[i\ldots i+L_i-1]$ is $\llr_i$.  Otherwise
  ($L_i = 0$), the letter $S[i]$ is a singleton, so $\llr_i$ does not
  exist.\qed
\end{proof}

Clearly, the left-ends of $\llr_1,\llr_2,\ldots,\llr_n$ strictly
increase as $1,2,\ldots, n$.  The next lemma shows the right-ends of
LLR's also monotonically increase.

\begin{lemma}
\label{lem:llr-length}
$|\llr_i|\leq |\llr_{i+1}|+1$, for every $i=1,2,\ldots, n-1$. 
\end{lemma}

\begin{proof}
 The claim is obviously correct for the cases when $\llr_i$ does not exist
  ($|\llr_i|=0$) or $|\llr_i|=1$, so we only consider the case when
  $|\llr_i| \geq 2$. Suppose $\llr_i = S[i\ldots j]$, $i < j$. It
  follows that $i+1 \leq j$. Since $S[i\ldots j]$ is a repeat, its
  substring $S[i+1\ldots j]$ is also a repeat. Note that $\llr_{i+1}$
  is the longest repeat substring starting from  position $i+1$, so
  $|\llr_{i+1}| \geq |S[i+1\ldots j]| = |\llr_i|-1$. \qed
\end{proof}

\begin{lemma}
\label{lem:lr-llr}
Every LR is an LLR.
\end{lemma}

\begin{proof}
Assume that $\lr_k=S[i\ldots j]$ is not an LLR. Note that $S[i\ldots
j]$ is a repeat starting from position $i$. If $S[i\ldots j]$ is not
an LLR, it means $S[i\ldots j]$ can be extend to some position
$j' > j$, so that $S[i\ldots j']$ is still a repeat and  also covers
position $k$. That says, $|S[i\ldots j']| > |S[i\ldots j]|$.
However, the contradiction is that $S[i\ldots j]$ is already the longest repeat
covering position $k$. \qed
\end{proof}

\section{Two Simple and Parallelizable Sequential Algorithms}
\label{sec:seq}
We know every LR is an LLR (Lemma~\ref{lem:lr-llr}), so the
calculation of a particular $\lr_k$ is actually a search for the
longest one among all LLR's that cover position $k$. Our discussion starts with the
finding of the leftmost LR for every position. In the end, an trivial
extension will be made to find all LR's for every string position.

\subsection{Use the raw LLR array}
We first calculate $\llr_i$, for $i=1,2,\ldots,n$, using
Lemma~\ref{lem:llr}, and save the result in an array $\llrr[1\ldots
n]$, where each $\llrr[i] = |\llr_i|$.  We call $\llrr[1\ldots n]$  the
\emph{raw} LLR array.  Because the rightmost LLR that covers position
$k$ is $\llr_k$ and the right boundaries of all LLR's monotonically
increase (Lemma~\ref{lem:llr-length}), the search for $\lr_k$ becomes
simply a walk from $\llr_k$ toward the left. The walk will stop when
it sees an LLR that does not cover position $k$ or it has reached the
left end of the LLRr array.  During this walk, we will record the
longest LLR that covers position $k$. Ties can be broken by storing
the leftmost such LLR. This yields the simple
Algorithm~\ref{algo:seq-1}, which outputs every LR as a $\langle
start, length\rangle$ tuple, representing the starting position and the
length of the LR.


\begin{algorithm}[t]
{\small
  \caption{Sequential finding of the leftmost $\lr_k$, $k=1,\ldots,n$, using the raw LLR array.}
\label{algo:seq-1}

\KwIn{The rank array and 
      the lcp array of the string $S$} 

\smallskip 

\tcc{Calculate $\llr_1,\llr_2,\ldots,\llr_n$.}

\lFor{$i=1,2,\ldots,n$}{
  $\llrr[i] \leftarrow \max\{\lcp[\rank[i]],\lcp[\rank[i]+1]\}$\tcp*{Length
    of $\llr_i$}

}

\smallskip

\tcc{Calculate $\lr_1,\lr_2,\ldots,\lr_n$.}
\For{$k=1,2,\ldots,n$}{

$\lr \leftarrow \langle -1,0\rangle$ \tcp*{$\langle start, end\rangle$: start and ending
  position of $\lr_k$.}

\For{$i = k$ down to $1$}{

  \If(\tcp*[f]{$\llr_i$ does not exist or does not cover $k$.}){$i+\llrr[i]-1<k$}{break\tcp*{Early stop}}
  \ElseIf{$\llrr[i] \geq \lr.length$}
   {$\lr\leftarrow \langle i,\llrr[i]\rangle$\;}
}
print $\lr$\;
}

}

\end{algorithm}


\begin{lemma}
\label{lem:seq-1}  
Given the rank and lcp arrays, Algorithm~\ref{algo:seq-1} can find 
the leftmost $\lr_k$ for every $k=1,2,\ldots,n$, using a total of 
$O(n)$ space and $O(\alpha n)$ time, where $\alpha$ is the average
number of LLR's that cover a string position. 
\end{lemma}

\begin{proof}
  (1) The time cost for the $\llrr$ array calculation is
  obviously $O(n)$.  The algorithm finds the LR of each of the
  $n$ string positions.  The average time cost for each LR
  calculation is bounded by the average number of walk steps,
  which is equal to the average number of LLR's that cover a
  string position. Altogether, the time cost is
  $O(\alpha n)$.  (2) The main memory space is used by the rank,
  lcp, and $\llrr$ arrays, each of which has $n$ integers. So
  altogether the space cost is $O(n)$ words.\qed
\end{proof}

\begin{theorem}
\label{thm:seq-1} 
We can find the leftmost $\lr_k$ for every $k=1,2,\ldots,n$, using a
total of $O(n)$ space and $O(\alpha n)$ time, where $\alpha$ is the
average number of LLR's that cover a string position.
\end{theorem}

\begin{proof}
  The suffix array of $S$ can be constructed using existing 
  $O(n)$-time and space algorithms (For example, \cite{KA-SA2005}). After the
  suffix array is constructed, the rank array can be trivially created
  using another $O(n)$ time and space.  We can then use the suffix
  array and the rank array to construct the lcp array using another
  $O(n)$ time and space~\cite{KLAAP01}. Combining with
  Lemma~\ref{lem:seq-1}, the claim in the theorem is proved. \qed
\end{proof}

\paragraph{Extension: find all LR's for every string position.}
As we have demonstrated in the example after Definition~\ref{def:lr},
a particular string position may be covered by multiple LR's, but
Algorithm~\ref{algo:seq-1} can only find the leftmost one. However,
extending it to find all LR's for every string position is trivial:
During each walk, we simply report all the longest LLR's that cover the
string position, of which we are computing the LR.  In order to do so, we will need to do the same walk
twice. The first walk is to find the length of the LR and the second
walk will actually report all the LR's.  We give the pseudocode of
this procedure in Algorithm~\ref{algo:seq-1-ext} in the appendix. This
algorithm certainly has another extra $O(\alpha)$ time cost on average
for each string position's LR calculation due to the extra walk, but it still gives a
total of $O(\alpha n)$ time cost and $O(n)$ space cost.

\begin{corollary}
\label{cor:seq-1-ext} 
We can find all LR's covering every position $k=1,2,\ldots,n$, using a
total of $O(n)$ space and $O(\alpha n)$ time, where $\alpha$ is the
average number of LLR's that cover a string position.
\end{corollary}

\paragraph{Comment:} 
(1) Algorithm~\ref{algo:seq-1} and~\ref{algo:seq-1-ext}
 are cache
friendly. Observe that the finding of every LR essentially is a linear
walk over a continuous chunk of the LLRr array. For real-world data, the
number of steps in every such walk is quite limited, as shown in the
upper rows of Table~\ref{tab:steps}. Note that the {\tt English}
dataset gives a much higher average number of walk steps, because the
data was synthesized by appending several real-world English texts
together, making many paragraphs appear several times.  Because of the
few walk steps needed for real-world data, the walking procedure can
thus be well cached in the L2 cache, whose size is around several MBs
in most nowadays desktops' CPU architecture, making our algorithm much
faster in practice. Note that the optimal $O(n)$
algorithm~\cite{IKX-repeat-CORR2015} uses a 2-table system to achieve its
optimality, which however has quite a pattern of random accessing the
different array locations during its run and thus is not cache
friendly.  We will demonstrate the comparison with more details in
Section~\ref{sec:exp}.
(2) Algorithm~\ref{algo:seq-1} and~\ref{algo:seq-1-ext} are
parallelizable in  shared-memory architecture. First, each LLR can be
calculated independently by a separate thread. After all LLR's are
calculated, each $LR$ can also be calculated independently by a
separate thread going through an independent walk. This enables us to
implement this algorithm on GPU, which supports massively parallel threads 
 using data parallelism.

\remove{
\begin{table}[t]
\centering
\def\0{\phantom{0}}
  \begin{tabular}{|l|l|c|c|c|}
\hline
LLR array type& \#walk steps  & DNA & English  & Protein \\
  \hline
&Minimum  & \0\0\0\0\,1  & \0\0\0\0\0\,1 & \0\0\0\0\,1\\

raw&Maximum & 14,836  & 109,394& 25,822\\

&Average  ($=\alpha$) & \0\0\0\,48  & \0\04,429& \0\0\,215\\

\hline

&Minimum  & \0\0\0\0\,1  & \0\0\0\0\0\,1  & \0\0\0\0\,1  \\

compact &Maximum   & \0\0\0\,14  & \0\0\0\0\,10  & \0\0\0\,35  \\

&Average  ($=\beta$)  &\0\0\0\0\,6  & \0\0\0\0\0\,2  & \0\0\0\0\,3  \\
\hline
   \end{tabular}
  \caption[test]{The number of walk steps in the LR's calculation using our solutions for
    several $50$MB files from Pizza\&Chili~\footnotemark}
  \label{tab:steps}
\end{table}
\footnotetext{\url{http://pizzachili.dcc.uchile.cl/texts.html}}
}

\subsection{Use  the Compact LLR Array}

Observe that an LLR can be a substring (suffix, more precisely) of
another LLR. For example, suppose $S={\tt ababab}$, then
$\llr_4=S[4\ldots 6]={\tt bab}$, which is a substring of
$\llr_3=S[3\ldots 6] = {\tt abab}$. We know every LR must be an LLR
(Lemma~\ref{lem:lr-llr}). So, if an $\llr_i$ is a substring of
another $\llr_j$, $\llr_i$ can never be the LR of any string position,
because every position covered by $\llr_i$ is also covered by at least
another longer LLR, $\llr_j$. 

\begin{definition}
We say an LLR is \emph{useless} if it is
a substring of another LLR; otherwise, it is \emph{useful}.
\end{definition}

Recall that in Algorithm~\ref{algo:seq-1} and~\ref{algo:seq-1-ext}, the
calculation of a particular $\lr_i$ is a search for the longest one
among all LLR's that cover position $i$. This search procedure is
simply a walk from $\llr_i$ toward the left until it sees an LLR that
does not cover position $i$ or reaches the left end of the LLRr array.
This search can be potentially sped up, if we have had all useless
LLR's eliminated before any search is performed. We will use a new
array LLRc, called the \emph{compact} LLR array, to store all the
useful LLR's in the ascending order of their left ends (as well as of
their right ends, automatically).

By Lemma~\ref{lem:llr-length}, we know if $\llr_{i-1}$ is not empty,
the right boundary of $\llr_i$ is on or after the right boundary of
$\llr_{i-1}$, for any $i\geq 2$.  So, we can construct the $\llrc$
array in one pass as follows.  We will calculate every $\llr_i$ using
Lemma~\ref{lem:llr}, for $i=1,2,\ldots,n$, and will eliminate every
$\llr_i$ if $|\llr_i| = 0$ or $|\llr_i| = |\llr_{i-1}|-1$.  Because of
the elimination of the useless LLR's, we will have to save each LLR as
a $\langle {\tt start,length}\rangle$ tuple, representing the starting
position and the length of the LLR, in the LLRc array.
Figure~\ref{fig:raw-compact} shows the geometric perspective of the
elements in an example $\llrr$ array and its corresponding $\llrc$
array, where every LLR is represented by a line segment whose start
and ending position represent the start and ending position of the
LLR.

Note that, in the LLRc array, any two LLR's share neither the same
left-end point (obviously) nor the same right-end point. In other
words, the left-end points of all useful LLR's strictly increase, and
so do their right-end points, i.e., all the elements in the LLRc
array have been sorted in the strict increasing order of their left-end (as well as
right-end) points. See Figure~\ref{fig:comp} for an example. 
Therefore, given a string position, we will be able
to find the leftmost useful LLR that covers that position using a binary
search over the LLRc array and the time cost for such a binary search is bounded by
$O(\log n)$. After that, we will simply walk along the LLRc array,
starting from the LLR returned by the binary search and toward the
right. The walk will stop when it sees an LLR that does not cover the
string position or it has reached the right end of the LLRc array.
During the walk, we will just report the longest LLR that covers the
given string position. Ties are broken by picking the leftmost such
longest LLR.  This leads to the Algorithm~\ref{algo:seq-2}.

\begin{figure}[t]
\centering
\begin{subfigure}{.5\textwidth}
  \centering
  \includegraphics[width=0.8\linewidth]{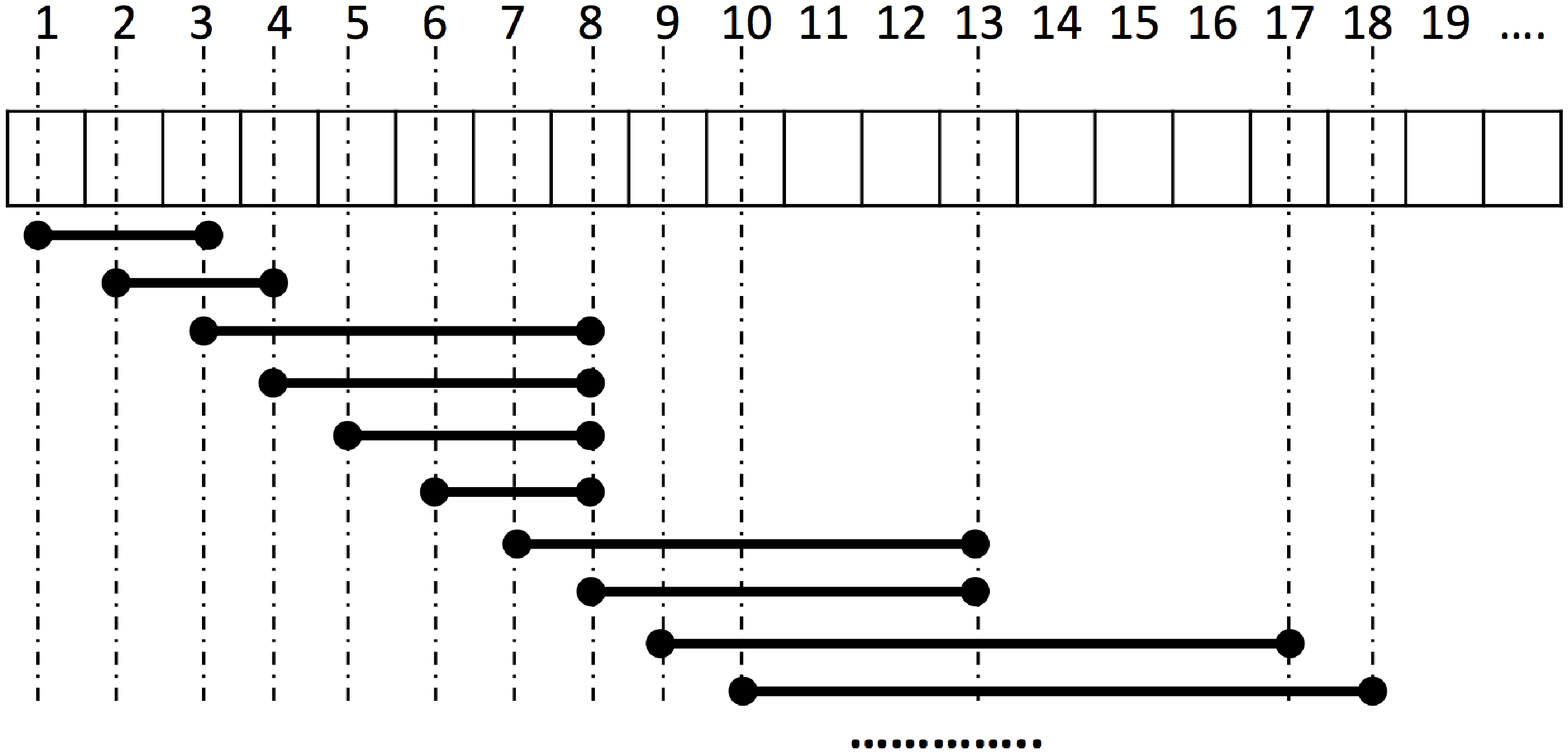}
  \caption{The raw LLR array}
  \label{fig:raw}
\end{subfigure}%
\begin{subfigure}{.5\textwidth}
  \centering
  \includegraphics[width=0.8\linewidth]{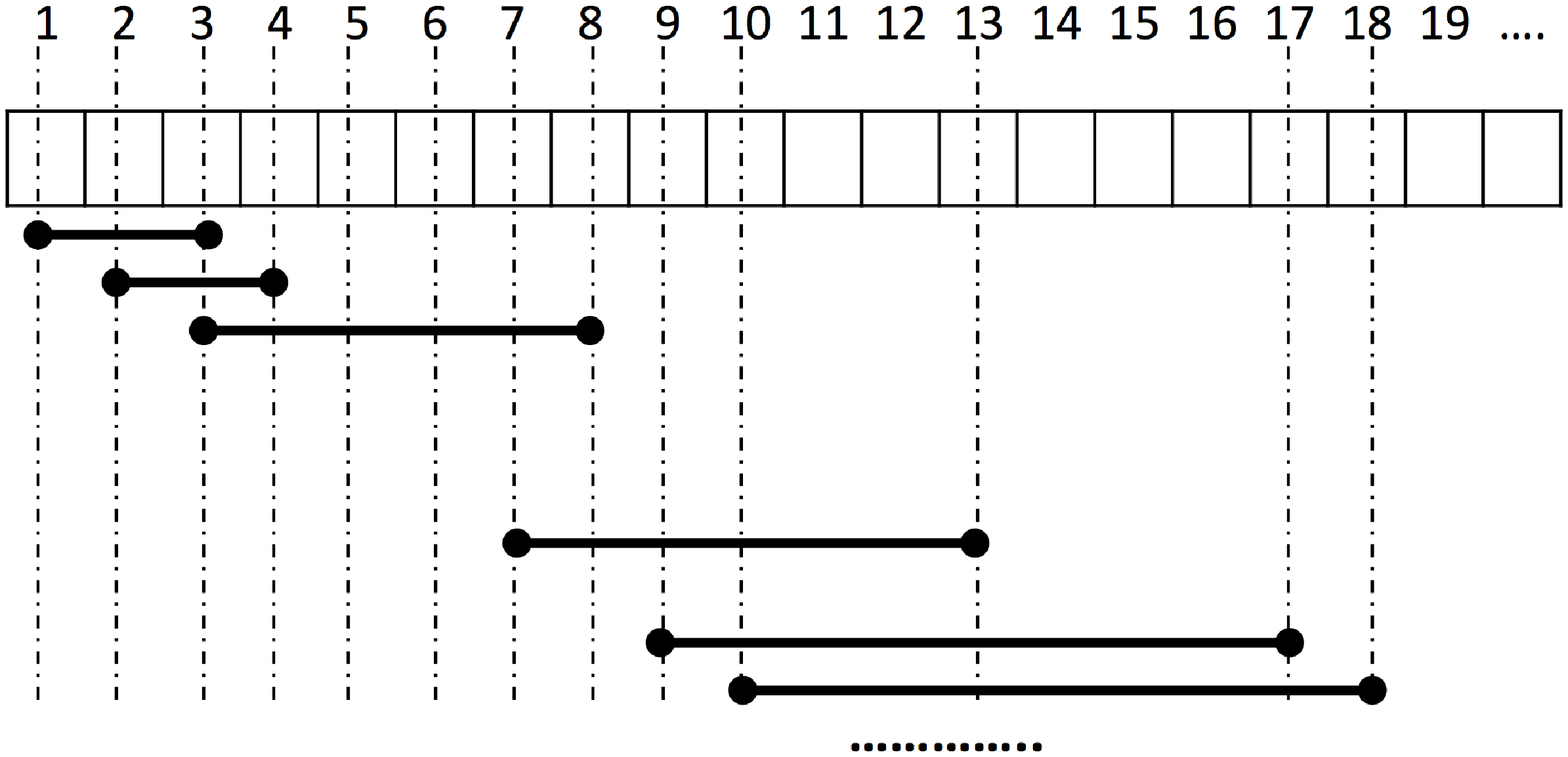}
  \caption{The compact LLR array}
  \label{fig:comp}
\end{subfigure}
  \caption{The geometric perspective of
an example raw LLR array and its corresponding compact LLR array}
\label{fig:raw-compact}
\end{figure}

\begin{lemma}
\label{lem:seq-2}  
Given the rank and lcp arrays, Algorithm~\ref{algo:seq-2} can find 
the leftmost $\lr_k$ for every $k=1,2,\ldots,n$, using a total of 
$O(n)$ space and $O(n (\log n + \beta))$ time, where $\beta$ is the average
number of useful LLR's that cover a string position. 
\end{lemma}

\begin{proof}
  (1) The time cost for the $\llrc$ array calculation is
  obviously $O(n)$ time.  The algorithm finds the LR of each of the
  $n$ string positions.  The average time cost for the calculation of the LR
  of one position includes the $O(\log n)$ time for the binary search and the
  time cost for the subsequent walk, which is bounded by the average number of
  useful LLR's that cover a string position.  Altogether, the time
  cost is $O(n (\log n + \beta))$.  (2) The main memory
  space is used by the rank, lcp, and $\llrc$ arrays. Each of the rank
  and lcp arrays has  $n$ integers. The $\llrc$ array has
  no more than $n$ pairs of integers. Altogether, the space cost
  $O(n)$ words. \qed
\end{proof}

\begin{theorem}
\label{thm:seq-2} 
We can find the leftmost $\lr_k$ for every $k=1,2,\ldots,n$, using a
total of $O(n)$ space and $O(n (\log n + \beta))$ time, where $\beta$
is the average number of useful LLR's that cover a string position.
\end{theorem}

\begin{proof}
  The suffix array of $S$ can be constructed by existing algorithms
  using $O(n)$ time and space (For example, \cite{KA-SA2005}). After the
  suffix array is constructed, the rank array can be trivially created
  using another $O(n)$ time and space.  We can then use the suffix
  array and the rank array to construct the lcp array using another
  $O(n)$ time and space~\cite{KLAAP01}. Combining the results in
  Lemma~\ref{lem:seq-2}, the theorem is proved. \qed
\end{proof}


\begin{algorithm}[t]
{\small
  \caption{Sequential finding of the leftmost $\lr_k$,
    $k=1,\ldots,n$, using the LLRc array.}
\label{algo:seq-2}

\KwIn{The rank array and 
      the lcp array of the string $S$} 

\smallskip 

\tcc{Calculate the compact LLR array.}
$j \leftarrow 1$; $prev\leftarrow 0$\;

\For{$i=1,2,\ldots,n$}{
  $L \leftarrow \max\{\lcp[\rank[i]],\lcp[\rank[i]+1]\}$\tcp*{Length
    of $\llr_i$}
  \lIf{$L> 0$ and $L \geq prev$}{
    $\llrc[j]\leftarrow \langle i, L\rangle$; $j\leftarrow j+1$\;
  }
  $prev\leftarrow L$\;
}
$size \leftarrow j-1$ \tcp*{Size of the $\llrc$ array. }

\smallskip

\tcc{Calculate $\lr_1,\lr_2,\ldots,\lr_n$.}
\For{$k=1,2,\ldots,n$}{

$\lr \leftarrow \langle -1,0\rangle$ \tcp*{$\langle start, end\rangle$: start and ending
  position of $\lr_k$.}

$start \leftarrow \textrm{BinarySearch}(\llrc,k)$\tcc{Return the
  smallest index of the $\llrc$ array element that covers position
  $k$, if such element exists; otherwise, return $-1$.}

\If{$start \neq -1$}{ 
  \For{$i = start \ldots size$}{
    
    \If(\tcp*[f]{$\llrc[i]$ does not cover
      $k$.}){$\llrc[i].start+\llrc[i].length-1<k$}
      {break\tcp*{Early stop}}
    \ElseIf{$\llrc[i].length > \lr.length$}
    {$\lr\leftarrow \llrc[i]$\;}
  }
}
  print $\lr$\;
}

}

\end{algorithm}


\paragraph{Extension: find all LR's for every string position.}
Algorithm~\ref{algo:seq-2} can also be trivially extended to find all
LR's for every string position by simply reporting all the longest LLR
that covers the position during every walk.  In order to do so, we
will need to walk twice for each string position. The first walk is to get the length of the LR
and the second walk will report all the actual LR's.  We give the
pseudocode of this procedure in Algorithm~\ref{algo:seq-2-ext} in the appendix.
This algorithm certainly has another extra $O(\beta)$ time
cost on average for each LR's calculation due to the extra walk, but
still gives a total of $O(n(\log n + \beta))$ time cost and $O(n)$
space cost.


\begin{corollary}
\label{cor:seq-2-ext} 
We can find all LR's of every string position $k=1,2,\ldots,n$, using a
total of $O(n)$ space and $O(n(\log n + \beta))$ time, where $\beta$ is the
average number of useful LLR's that cover a string position.
\end{corollary}

\paragraph{Comment:} 
(1) The binary searches that are involved in
Algorithms~\ref{algo:seq-2} and~\ref{algo:seq-2-ext} are not cache
friendly.  However, compared with
Algorithms~\ref{algo:seq-1} and \ref{algo:seq-1-ext},
Algorithms~\ref{algo:seq-2} and \ref{algo:seq-2-ext} on average have much
fewer steps (the $\beta$ value) in each walk due to the elimination of
the useless LLR's (see the bottom rows of Table~\ref{tab:steps}).
This makes Algorithms~\ref{algo:seq-2} and \ref{algo:seq-2-ext} much
better choices rather than Algorithms~\ref{algo:seq-1}
and~\ref{algo:seq-1-ext} for run environments that have small cache
size. Such run environments include the GPU architecture, where the
cache size for each thread block is only several KBs.  We will
demonstrate this claim with more details in Section~\ref{sec:exp}.
(2) With more care in the design, Algorithms~\ref{algo:seq-2}
and~\ref{algo:seq-2-ext} are also parallelizable in 
shared-memory architecture (SMA), which is described in the next
Section.

\section{Parallel Implementation on GPU}
\label{sec:para}

In this section, we describe the GPU version of
Algorithms~\ref{algo:seq-1} and \ref{algo:seq-2} and their extensions
(Algorithms~\ref{algo:seq-1-ext} and~\ref{algo:seq-2-ext}).
%
%
%
After we construct the SA, Rank, and LCP arrays on the host CPU~\footnote{
  The SA, Rank, and LCP arrays can also be constructed in parallel on
  GPU~\cite{osipov-spire2012,DK-ppopp2013}, but due to the
  unavailability of the source code or executables from the authors
  of~\cite{osipov-spire2012,DK-ppopp2013}, we choose to construct
  these arrays on the host CPU, without affecting the demonstration of
  the performance gains by our algorithms.}, we transfer
the Rank array and the LCP array to the GPU device memory. 
We start with the calculation of the raw LLR array in parallel.

\paragraph{Compute the raw LLR array.}
After the LCP and Rank arrays are loaded into GPU memory, we launch a
CUDA kernel to compute the raw LLR array on GPU device using massively
parallel
threads, 
as illustrated in Figure~\ref{fig:bigPic}. Each thread $t_i$ on the
device computes a separate element $\llrr[i]=|\llr_i|$ using the
following equation from Lemma~\ref{lem:llr}.
$$\llr[i]=\max\{\lcp[\rank[i], \lcp[\rank[i] + 1]\}$$
Since
each $\llr_i$ must start with string position $i$, we only need to
save the length of each $\llr_i$ in $\llrr[i]$.
%
After creating the raw LLR array, we have two options, which
in turn lead to two different parallel solutions: using the raw LLR
array or the compact LLR array.

\subsection{Compute LR's using the raw LLR array}

%
%
%
%

\remove{ First, we condense the original LLR array by removing the
  redundant LLR[i], the line segments starting with i, with a length
  of LLR[i], which has been completely covered by a previous segment
  LLR[j], where $j < i$.
  We obtain a compact LLR array, named as LLRc array, with $LLRc[i]$
  represented as a tuple of (start, length).  Second, to compute the
  LR, we use the original LLR array as it is, the raw LLR array with a
  lot of redundant information in it.  Let us first discuss the option
  one.  }

The parallel implementation of Algorithm~\ref{algo:seq-1} using the
raw LLR array is straightforward, as presented by the left branch of
Figure~\ref{fig:bigPic}.  With the raw LLR array returned from the
previous kernel launch on the GPU device, we launch a second kernel
for LR calculation.  Each CUDA thread $t_i$ on the device is to find
$\lr_i$ by performing a linear walk in the LLRr array, starting at
$\llrr[i]$ toward the left. The walk continues until it finds an LLRr
array element that does not cover position $i$ or has reached the left
end of the LLRr array. The leftmost or all $\lr_i$ can be reported
during the walk, as discussed in Algorithm~\ref{algo:seq-1}.
Note that in this search, each CUDA thread checks a chunk of
contiguous elements in the LLRr array and this can be cache-efficient.

Taking the calculation of $LR_{10}$ using the raw $\llr$ array shown
in Figure~\ref{fig:raw} as an example.  The corresponding CUDA thread
$t_{10}$ searches a contiguous chunk of the LLRr array starting from
index $10$ down to left in the $\llrr$ array.  We do not search the
LLRr elements that are to the right of index $10$, because these
elements definitely do not cover position $10$ according to the
definition of LLR.  In particularly, thread $t_{10}$ goes through
$\llrr[10]$, $\llrr[9]$, $\llrr[8]$ and $\llrr[7]$ to find the longest
one among the four of them as $\lr_{10}$. Thread $t_{10}$ stops the
search at LLRr position $6$, because $\llrr[6]$ and all LLR's to its
left do not cover position $10$ (Lemma~\ref{lem:llr-length}).


\begin{figure}[t]
\begin{minipage}{.5\textwidth}
\centering
\scalebox{0.55} 
{
\begin{pspicture}(0,-3.2996354)(15.504896,3.3396354)
\usefont{T1}{ptm}{m}{n}
\rput(6.5497394,2.9903646){\psframebox[linewidth=0.04]{SA, Rank, LCP array construction on CPU}}
\usefont{T1}{ptm}{m}{n}
\rput(7.434896,1.6103646){\psframebox[linewidth=0.04]{Compute: $\llrr[i] = \max\{\lcp[\rank[i]], \lcp[\rank[i]+1]\}$, for $i = 1, 2, \ldots, n$}}
\usefont{T1}{ptm}{m}{n}
\rput(3.5209897,0.27036458){\psframebox[linewidth=0.04]{Use the raw LLR array}}
\usefont{T1}{ptm}{m}{n}
\rput(10.315677,0.3503646){Compact the raw LLR array into}
\usefont{T1}{ptm}{m}{n}
\rput(10.502865,-0.08963542){LLRc array of <start,length> tuples}
\usefont{T1}{ptm}{m}{n}
\rput(2.7064583,-1.2496355){Compute every LR[k] by}
\usefont{T1}{ptm}{m}{n}
\rput(2.9284897,-1.7496355){linearly scanning LLRr[i...k], }
\usefont{T1}{ptm}{m}{n}
\rput(9.703333,-1.2696354){Compute LR[k], using binary search to find}
\usefont{T1}{ptm}{m}{n}
\rput(10.550208,-1.7496355){$i = \min\{t | \llrc[t].start + \llrc[t].length - 1 \geq k\}$,}
\usefont{T1}{ptm}{m}{n}
\rput(9.370364,-2.2496355){then a linear walk through LLRc[i...j],}
\usefont{T1}{ptm}{m}{n}
\rput(2.940052,-2.3096354){$i = \min\{j | j + \llr[j] \geq k\}$}
\usefont{T1}{ptm}{m}{n}
\rput(10.470052,-2.7496355){$j = \max\{t | \llrc[t].start + \llrc[t].length - 1 \geq k\}$}
\psframe[linewidth=0.04,dimen=outer](5.305521,-0.89963543)(0.74552083,-2.6996355)
\psframe[linewidth=0.04,dimen=outer](13.3655205,0.7203646)(7.6655207,-0.3996354)
\psframe[linewidth=0.04,dimen=outer](14.3255205,-0.89963543)(6.6855206,-3.1396353)
\psline[linewidth=0.06cm,arrowsize=0.05291667cm 2.0,arrowlength=1.4,arrowinset=0.4]{->}(6.265521,2.6003647)(6.285521,2.0803645)
\psline[linewidth=0.04cm,arrowsize=0.05291667cm 2.0,arrowlength=1.4,arrowinset=0.4]{->}(6.345521,1.2003646)(5.2255206,0.46036458)
\psline[linewidth=0.04cm,arrowsize=0.05291667cm 2.0,arrowlength=1.4,arrowinset=0.4]{->}(6.425521,1.1803646)(7.565521,0.4803646)
\psline[linewidth=0.04cm,arrowsize=0.05291667cm 2.0,arrowlength=1.4,arrowinset=0.4]{->}(2.8455207,-0.15963541)(2.8655207,-0.83963543)
\psline[linewidth=0.04cm,arrowsize=0.05291667cm 2.0,arrowlength=1.4,arrowinset=0.4]{->}(10.505521,-0.43963543)(10.525521,-0.8796354)
\usefont{T1}{ptm}{m}{n}
\rput(4.95349,0.9303646){option 1}
\usefont{T1}{ptm}{m}{n}
\rput(7.832083,0.9703646){option 2}
\psframe[linewidth=0.04,linestyle=dashed,dash=0.16cm 0.16cm,dimen=outer](14.505521,2.0403645)(0.26552084,-3.2996354)
\usefont{T1}{ptm}{m}{n}
\rput(11.923959,2.3103645){Parallelized on GPU}
\end{pspicture} 
}
  \captionof{figure}{Overview of the GPU Implementation}
  \label{fig:bigPic}
\end{minipage}%
\begin{minipage}{.5\textwidth}
  \centering
\vspace*{11.2mm}
\scalebox{0.7} 
{
\begin{pspicture}(0,-1.8468945)(11.217891,1.8491992)
\usefont{T1}{ptm}{m}{n}
\rput(0.50439453,-0.9841992){<1,3>}
\psframe[linewidth=0.03,dimen=outer](1.04,-0.7691992)(0.0,-1.2091992)
\usefont{T1}{ptm}{m}{n}
\rput(1.5443945,-0.9841992){<4,1>}
\psframe[linewidth=0.03,dimen=outer](2.08,-0.7691992)(1.04,-1.2091992)
\usefont{T1}{ptm}{m}{n}
\rput(2.5843945,-0.9841992){<5,3>}
\psframe[linewidth=0.03,dimen=outer](3.12,-0.7691992)(2.08,-1.2091992)
\usefont{T1}{ptm}{m}{n}
\rput(3.6243944,-0.9841992){<8,1>}
\psframe[linewidth=0.03,dimen=outer](4.16,-0.7691992)(3.12,-1.2091992)
\usefont{T1}{ptm}{m}{n}
\rput(0.44776368,-1.5841992){\ \ \ 1}
\psframe[linewidth=0.03,dimen=outer](1.04,-1.3691993)(0.0,-1.8091992)
\usefont{T1}{ptm}{m}{n}
\rput(1.4877636,-1.5841992){\ \ \ 1}
\psframe[linewidth=0.03,dimen=outer](2.08,-1.3691993)(1.04,-1.8091992)
\usefont{T1}{ptm}{m}{n}
\rput(2.5277636,-1.5841992){\ \ \ 1}
\psframe[linewidth=0.03,dimen=outer](3.12,-1.3691993)(2.08,-1.8091992)
\usefont{T1}{ptm}{m}{n}
\rput(3.5822656,-1.5841992){\ \ \ 2}
\psframe[linewidth=0.03,dimen=outer](4.16,-1.3691993)(3.12,-1.8091992)
\usefont{T1}{ptm}{m}{n}
\rput(4.6147947,-1.5841992){\ \ \ 3}
\psframe[linewidth=0.03,dimen=outer](5.2,-1.3691993)(4.16,-1.8091992)
\usefont{T1}{ptm}{m}{n}
\rput(5.6547947,-1.5841992){\ \ \ 3}
\psframe[linewidth=0.03,dimen=outer](6.24,-1.3691993)(5.2,-1.8091992)
\usefont{T1}{ptm}{m}{n}
\rput(6.694795,-1.5841992){\ \ \ 3}
\psframe[linewidth=0.03,dimen=outer](7.28,-1.3691993)(6.24,-1.8091992)
\usefont{T1}{ptm}{m}{n}
\rput(7.743496,-1.5841992){\ \ \ 4}
\psframe[linewidth=0.03,dimen=outer](8.32,-1.3691993)(7.28,-1.8091992)
\usefont{T1}{ptm}{m}{n}
\rput(0.44776368,1.0558008){\ \ \ 1}
\psframe[linewidth=0.03,dimen=outer](1.04,1.2708008)(0.0,0.8308008)
\usefont{T1}{ptm}{m}{n}
\rput(1.5033203,1.0558008){\ \ \ 0}
\psframe[linewidth=0.03,dimen=outer](2.08,1.2708008)(1.04,0.8308008)
\usefont{T1}{ptm}{m}{n}
\rput(2.5433204,1.0558008){\ \ \ 0}
\psframe[linewidth=0.03,dimen=outer](3.12,1.2708008)(2.08,0.8308008)
\usefont{T1}{ptm}{m}{n}
\rput(3.5677636,1.0558008){\ \ \ 1}
\psframe[linewidth=0.03,dimen=outer](4.16,1.2708008)(3.12,0.8308008)
\usefont{T1}{ptm}{m}{n}
\rput(4.607764,1.0558008){\ \ \ 1}
\psframe[linewidth=0.03,dimen=outer](5.2,1.2708008)(4.16,0.8308008)
\usefont{T1}{ptm}{m}{n}
\rput(5.6633205,1.0558008){\ \ \ 0}
\psframe[linewidth=0.03,dimen=outer](6.24,1.2708008)(5.2,0.8308008)
\usefont{T1}{ptm}{m}{n}
\rput(6.7033205,1.0558008){\ \ \ 0}
\psframe[linewidth=0.03,dimen=outer](7.28,1.2708008)(6.24,0.8308008)
\usefont{T1}{ptm}{m}{n}
\rput(7.7277637,1.0558008){\ \ \ 1}
\psframe[linewidth=0.03,dimen=outer](8.32,1.2708008)(7.28,0.8308008)
\usefont{T1}{ptm}{m}{n}
\rput(0.4547949,0.45580077){\ \ \ 3}
\psframe[linewidth=0.03,dimen=outer](1.04,0.6708008)(0.0,0.23080078)
\usefont{T1}{ptm}{m}{n}
\rput(1.5022656,0.45580077){\ \ \ 2}
\psframe[linewidth=0.03,dimen=outer](2.08,0.6708008)(1.04,0.23080078)
\usefont{T1}{ptm}{m}{n}
\rput(2.5277636,0.45580077){\ \ \ 1}
\psframe[linewidth=0.03,dimen=outer](3.12,0.6708008)(2.08,0.23080078)
\usefont{T1}{ptm}{m}{n}
\rput(3.5677636,0.45580077){\ \ \ 1}
\psframe[linewidth=0.03,dimen=outer](4.16,0.6708008)(3.12,0.23080078)
\usefont{T1}{ptm}{m}{n}
\rput(4.6147947,0.45580077){\ \ \ 3}
\psframe[linewidth=0.03,dimen=outer](5.2,0.6708008)(4.16,0.23080078)
\usefont{T1}{ptm}{m}{n}
\rput(5.662266,0.45580077){\ \ \ 2}
\psframe[linewidth=0.03,dimen=outer](6.24,0.6708008)(5.2,0.23080078)
\usefont{T1}{ptm}{m}{n}
\rput(6.6877637,0.45580077){\ \ \ 1}
\psframe[linewidth=0.03,dimen=outer](7.28,0.6708008)(6.24,0.23080078)
\usefont{T1}{ptm}{m}{n}
\rput(7.7277637,0.45580077){\ \ \ 1}
\psframe[linewidth=0.03,dimen=outer](8.32,0.6708008)(7.28,0.23080078)
\psline[linewidth=0.04cm,arrowsize=0.05291667cm 2.0,arrowlength=1.4,arrowinset=0.4]{->}(0.52,0.15080078)(0.52,-0.6891992)
\psline[linewidth=0.04cm,arrowsize=0.05291667cm 2.0,arrowlength=1.4,arrowinset=0.4]{->}(3.6,0.11080078)(1.56,-0.6891992)
\psline[linewidth=0.04cm,arrowsize=0.05291667cm 2.0,arrowlength=1.4,arrowinset=0.4]{->}(4.6,0.11080078)(2.76,-0.6891992)
\psline[linewidth=0.04cm,arrowsize=0.05291667cm 2.0,arrowlength=1.4,arrowinset=0.4]{->}(7.48,0.11080078)(3.68,-0.6891992)
\usefont{T1}{ptm}{m}{n}
\rput(9.256045,1.0558008){Flag array}
\usefont{T1}{ptm}{m}{n}
\rput(9.597803,0.45580077){raw LLR array}
\usefont{T1}{ptm}{m}{n}
\rput(5.8527927,-0.9441992){compact LLR array}
\usefont{T1}{ptm}{m}{n}
\rput(9.845957,-1.6241993){Prefix\_Sum array}
\usefont{T1}{ptm}{m}{n}
\rput(0.51145506,1.6558008){$t_1$}
\usefont{T1}{ptm}{m}{n}
\rput(1.551455,1.6558008){$t_2$}
\usefont{T1}{ptm}{m}{n}
\rput(2.6314552,1.6558008){$t_3$}
\usefont{T1}{ptm}{m}{n}
\rput(3.6714551,1.6558008){$t_4$}
\usefont{T1}{ptm}{m}{n}
\rput(4.711455,1.6558008){$t_5$}
\usefont{T1}{ptm}{m}{n}
\rput(5.7514553,1.6558008){$t_6$}
\usefont{T1}{ptm}{m}{n}
\rput(6.831455,1.6558008){$t_7$}
\usefont{T1}{ptm}{m}{n}
\rput(7.871455,1.6558008){$t_8$}
\usefont{T1}{ptm}{m}{n}
\rput(9.437392,1.6158007){GPU threads}
\end{pspicture} 
}
  \captionof{figure}{LLR compaction on GPU}
  \label{fig:compact}
\end{minipage}
\end{figure}


\subsection{Compute LR's using the compact LLR array}

\subsubsection{LLR Compaction.}
The right branch of Figure~\ref{fig:bigPic} shows the second option in
computing LR's on GPU. That is to use the compact LLR array.  We first
create the compact LLR array, named as LLRc, from the raw LLR array,
which has been created and preserved on the device memory.  To avoid
the expensive data transfer between the host and the device and to
achieve more parallelism, we perform the LLR array compaction on the
GPU device in parallel.  We launch three CUDA kernels to perform the
compaction, denoted as $\mathcal{K}_1$, $\mathcal{K}_2$, and
$\mathcal{K}_3$.  As shown in Figure~\ref{fig:compact}, after the LLRr
array is constructed on the device, we first launch kernel
$\mathcal{K}_1$ to compute a flag array $Flag[1\ldots n]$ in parallel,
where the value of each element $Flag[i]$ is assigned by a separate
thread $t_i$ as follows: (1) $Flag[1]=1$ iff $\llrr[1]>0$.  (2)
$Flag[i] = 1$, iff $\llrr[i]>0$ and $\llrr[i] \geq \llrr[i-1]$, for
$i=2,3,\ldots,n$. $Flag[i]=0$ means $\llr_i$ is useless and thus can
be eliminated.

After the Flag array is constructed from kernel $\mathcal{K}_1$, we launch kernel
$\mathcal{K}_2$ to calculate the prefix sum of the Flag array on the device:
$\ps[i]=\sum_{j=1}^{i}Flag[j]$.  We modify the prefix sum function
provided by the CUDA toolkit for this purpose.

With the prefix sum array and the Flag array, we launch kernel $\mathcal{K}_3$ to
copy the useful LLRr array elements into the LLRc array, as illustrated in
Figure~\ref{fig:compact}.  Each thread $t_i $ on the device moves in
parallel the $\llr_i$ to an unique destination $\llrc[\ps[i]]$, if
$Flag[i] = 1$. That is, $\llrc[\ps[i]] = \langle i, \llrr[i]\rangle$,
if $Flag[i] = 1$.  Each element in the $\llrc$ array is a useful LLR and
is represented by a tuple of $\langle start, length\rangle$, the
start and ending position of the LLR. 

\remove{
{\bf (Please Clarify the context of this paragraph ...)} There are
two points worth mentioning here. First, we could merge kernels
$\mathcal{K}_1$ with the previous kernel that calculates the raw LLR
array.  But we find the merge has negligible effects in performance
improvement, mainly because both kernels only use the data that have
been stored on the device. Second, as we consider each LLR in the LLRc
array as a line segment, the start positions of these LLR's are still
sorted, but not contiguous any more because we have discarded the
useless LLR's. As shown in Figure~\ref{fig:comp}, the LLR's starting at
position 4, 5, 6, 8 are useless, thus are discarded. Therefore, for
example, the LLR starting at position 10 is not necessarily stored at
the index 10 in the LLRc array.}


\remove{ Back to Figure~\ref{fig:bigPic}, we focus on the first
  parallel solution, which uses a compact LLRc array.  As we have
  suggested, each item $LLRc[i]$ represents a line segment, by using a
  tuple \{LLRc[i].start, LLRc[i].length\}.  After compaction, line
  segments in LLRc array are ordered by both start position and end
  position of the line segments, which has been proved in previous
  section \ref{}; }

\subsubsection{Compute LR's.}
After the LLRc array is prepared, we calculate the LR for
every string position in parallel.  Recall that the calculation of
each $\lr_k$, for each $k=1,2,\ldots,n$, is a search for the longest
useful LLR that covers position $k$. We also know all these relevant
LLR's that we need to search comprise a continuous chunk of the
LLRc array. The start position of the chunk can be found using a
binary search as we have explained in the discussion of
Algorithm~\ref{algo:seq-2}.  After that, a simple linear walk toward
the right is performed. The walk continues until it finds an
LLRc array element that does not cover position $k$ or has reached the
right end of the LLRc array. 

To compute the LR's using the LLRc array, we launch another CUDA
kernel, in which each CUDA thread $t_k$ first performs a binary search
to find the start position of the linear walk and then walk through
the relevant LLRc array elements to find either all LR's or a single LR
covering position $k$.

Referring to Figure~\ref{fig:comp}, we take the LR calculation covering the
string position $9$ as an example. Recall that we have discarded
all useless LLR's in the LLRc array, so the LLRc array element at
index $9$ is not necessarily the rightmost LLR that cover string
position $9$. Therefore, we have to perform a binary search to locate
that leftmost LLRc array element by taking advantage of the nice
property of the LLRc array that both the start and ending positions of
all LLR's in it are strictly increasing.
After thread $t_{9}$ locates the LLRc element $\llrc[4]$, the leftmost useful LLR that
covers the string position $9$, it performs a linear walk toward the
right. The walk will continue until it meets $\llrc[6]$, which does 
not cover position $9$. Thread $t_9$ will return the longest ones
among $\llrc[4\ldots 6]$ as $\lr_9$.

\subsection{Advantages and Disadvantages: $\llrr$ vs.\ $\llrc$}

When the raw LLR array is used, the algorithm is straightforward and
easy to implement, because there is no needs to perform the LLR
compaction on the device. However, with a raw LLR array, we could
have a large number of useless LLR's in the raw LLR array, especially
when the average length of the longest repeats is quite large. For that reason, the
subsequent linear walk for each CUDA thread can take many steps,
making the overall search performance worse. 

In contrast, under a compact LLR array, we have to perform the LLR
compaction, which involves data coping and requires extra memory usage
for the $Flag$ and the prefix sum array on the device. In addition, a
binary search, which is not present with a raw LLR array, is required
to locate the first LLR for the linear walk. The advantage of a
compact LLR array is that we remove the useless LLR's and dramatically
shorten the linear walk distance.We provide more analysis and
comparison between these two solutions in the experiment section.

\section{Experimental Study}
\label{sec:exp}


\paragraph{Experiment Environment Setup.}
We conducted our experiments  on a computer running GNU/Linux
with a kernel version 3.2.51-1.  The computer is equipped with an
Intel Xeon 2.40GHz E5-2609 CPU with 10MB Smart Cache and has
 16GB RAM.
 We used a GeForce GTX 660 Ti GPU for our parallel tests. 
 The GPU consists of 1344 CUDA cores and 2GB of RAM memory.  The GPU
 is connected with the host computer with a PCI Express 3.0 interface.
 We install CUDA toolkit 5.5 on the host computer.  We use {\tt C} to
 implement our sequential algorithms and use {\tt CUDA C} to implement
 our parallel solutions
 on the GPU, using {\tt gcc 4.7.2} with {\tt -O3}
 option and {\tt nvcc V5.5.0} as the compilers.
We test our algorithms on real-world datasets including biological and
non-biological data downloaded from the Pizza\&Chili Corpus.  The
datasets we used are the three $50$MB {\tt DNA}, {\tt English}, and
{\tt Protein} pure ASCII text files, each of which thus represents 
a string of 
$50\times 1024 \times 1024$ characters.

\paragraph{Measurements.} 

We measured the average time cost of three runs of our program.
In order to better highlight the comparison of the algorithmics
between the old and our new solutions, we did not include the time
cost for the I/O operations that save the results.  
For the same purpose,
we also did not include the time cost for the SA, Rank, and LCP array
constructions, because in both the old and our new
solutions, these auxiliary data structures are constructed based on
the same best suffix array construction code available on the
Internet~\footnote{\url{http://code.google.com/p/libdivsufsort/}}.  
Our source code for this work is also available on website.\footnote{
\url{http://penguin.ewu.edu/~bojianxu/publications}}

\subsection{Time}
In the top three charts of Figure~\ref{fig::time-size}, using three
datasets, we compare different algorithms that return only the
leftmost LR for every string position of the input data.  In the
bottom three charts, we present the performance of our algorithms that
are able to find $all$ LR's for every string position. We compare our
new algorithms with the existing optimal sequential
algorithm~\cite{IKX-repeat-CORR2015}, which can only find the leftmost LR
for every string position. Table~\ref{tab:speedup} summarizes the
speedup of our algorithms against the old optimal algorithm.
From  experiments, we are able to make the following observations. 



\paragraph{Sequential algorithms on CPU.} 
Our new sequential algorithm using the raw LLR  is consistently
faster than the old optimal algorithm by a factor of $1.97$--$3.44$, while our
new sequential algorithm that uses the compact LLR array is consistently
slower. This observation is true in both finding the leftmost LR and
all LR's.  (Please note that the old optimal algorithm always finds the
leftmost LR only.)

On the host CPU, three dominating factors contribute to the better
performance of algorithms using a raw LLR array rather than using a
compact LLR array. First, although the compact LLR array can still be
constructed in one pass, but the construction involves a lot more
computational steps than those needed in the construction of the raw
LLR array.  Second, sequential algorithms that use a compact LLR array
require a binary search in order to locate the starting position of
the subsequent linear walk in the calculation of every LR.  However,
binary searches are not required if we work with a raw LLR array. As
it is known, binary search over a large array is not cache
friendly. Through profiling, we observe that the binary search
operations consume from $63\%$ to $73\%$ of the total execution time.  Third,
even though for some datasets the search range size (or the number of
walk steps) with a raw LLR array could be $10,000$ times larger than
that using a compact LLR array, as shown in table \ref{tab:steps}, the
L2 cache ($10$MB) of the host CPU is large enough to cache the range
of $contiguous$ LLR's that each linear walk needs to go through.  Such
efficient data caching helps all walks take less than a total of 100
milliseconds on the host CPU, accounting for less than $5\%$ of the
total execution time, even with the raw LLR array.  In other words,
given a large cache memory, the number of walk steps is no longer a
dominating factor in the overall performance.


\paragraph{Parallel algorithms on GPU.} 
Our new parallel algorithm on GPU using the compact LLR array is
consistently faster than its counterpart that uses the raw LLR array,
which is consistently faster than the old optimal algorithm by a
factor of $8.32$--$14.62$ in finding the leftmost LR and
$6.36$--$10.35$ in finding all LR's.

Unlike the sequential algorithm on the host CPU, the performance of
the parallel algorithm on the GPU device is dominated by the number of
LLR's (the number of walk steps) that each walk will go through.  As we
profile our GPU implementation, we observe that with the raw LLR
array, all linear walks on the GPU take roughly a total of eight
seconds for the $English$ dataset.  But, the walks take
roughly $70$ milliseconds only if using a compact LLR array on the GPU.
This is because: (1) the small GPU L2 cache ($384$KB shared by all
streaming multiprocessors) cannot host as many LLR's as what the CPU
L2 cache (10MB) can host, resulting in more cache-read misses and more
expensive global memory accesses. (2) The number of walk steps with a
compact LLR array is less than that with a raw LLR array by a factor
of up to four orders of magnitude (see Table~\ref{tab:steps}).  (3) The
extra time cost for the LLR compaction that is needed when using the
compact LLR array become much less significant in the total execution time on GPU.  
On the host CPU, our sequential solution
takes roughly $1.3$ seconds to perform the LLR compaction for datasets
of $50$MB and accounts for $20\%$ of the total time
cost on average. However, it takes less than $30$ milliseconds on the GPU,
accounting for only $9.5\%$ of the entire time cost.  We achieve more
than $40$ times speedup in the LLR compaction by utilizing GPU device.

The first two reasons above are reassured by the experimental results
regarding the {\tt English} dataset, which we purposely chose to use.
The {\tt English} file is synthesized by simply concatenating several
English texts, and thus the text has many repeated paragraphs, which
in turn creates many \emph{useless} LLR's in the data.  In this case,
with the raw LLR array, each walk will have a large number of steps due to
such useless LLR's.  However, after we compact the raw LLR array, the
number of walk steps can be significantly reduced
(Table~\ref{tab:steps}) and consequently the GPU code's performance is
significantly improved (Figure~\ref{fig::time-size}).

\begin{figure}[t]
\begin{center}
\begin{tabular}{ccc}
\includegraphics[scale=.62]{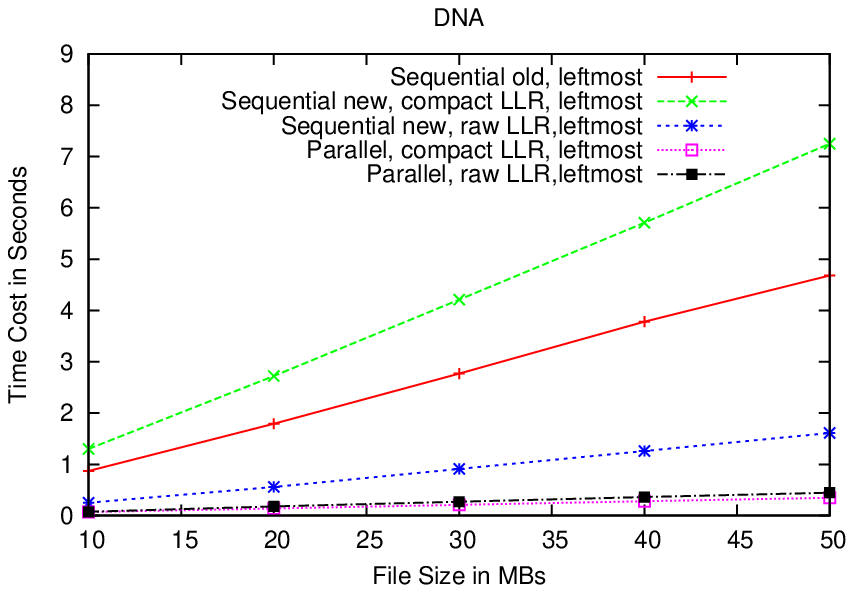} &
\includegraphics[scale=.62]{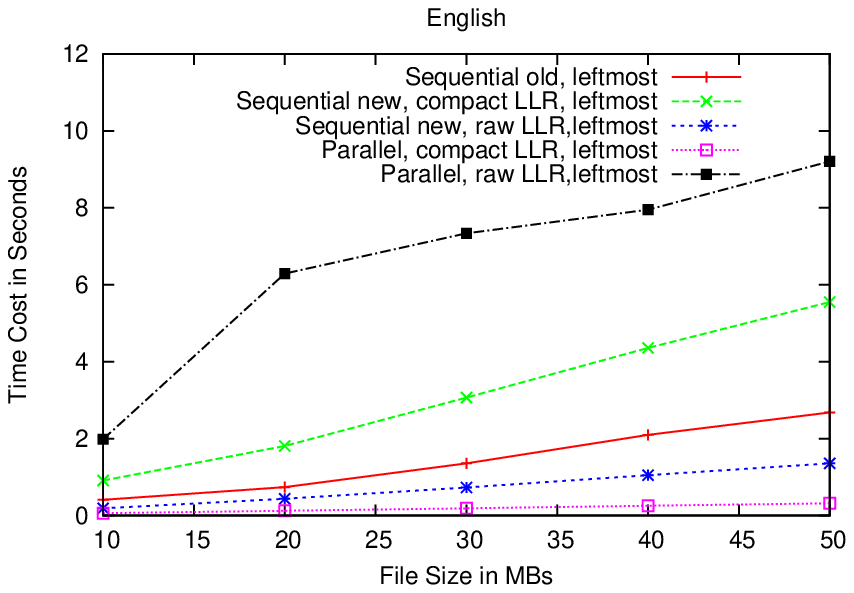} & 
\includegraphics[scale=.62]{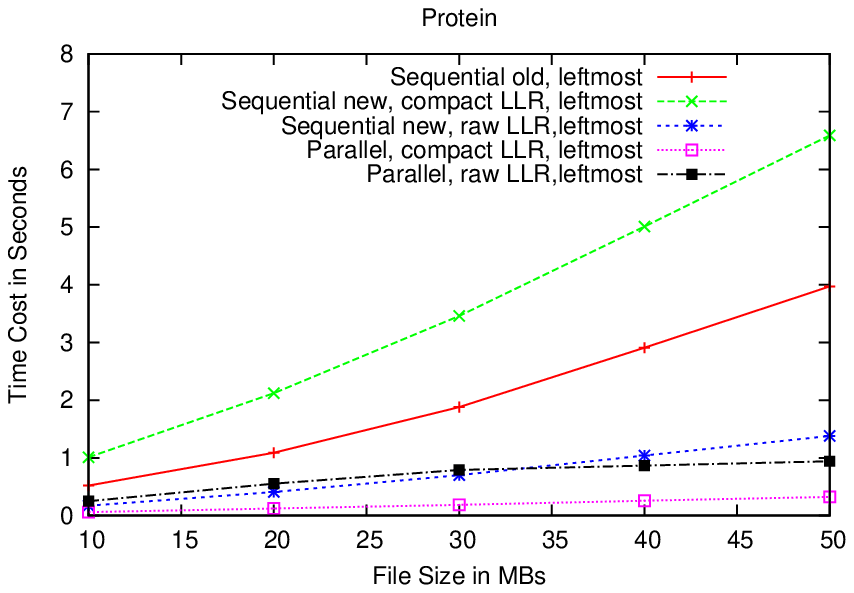}  \\
\includegraphics[scale=.62]{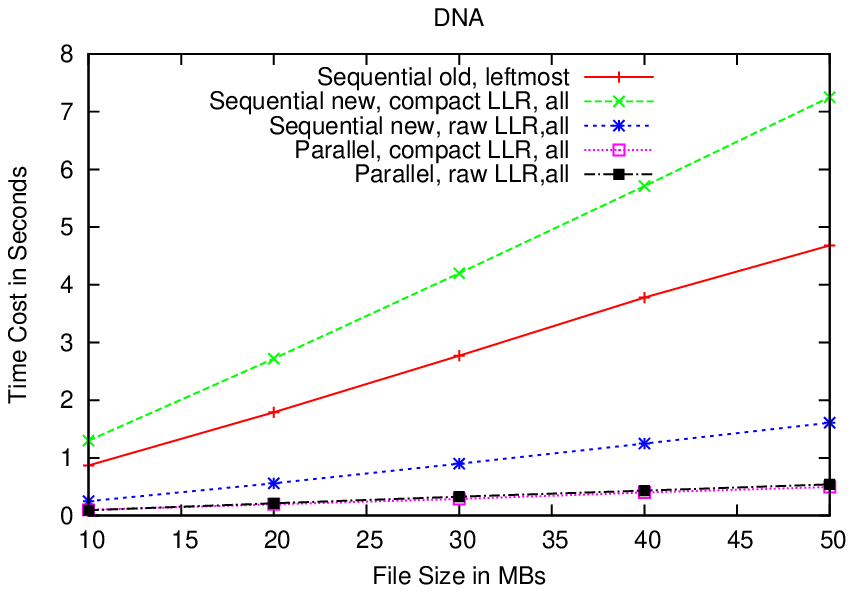} &
\includegraphics[scale=.62]{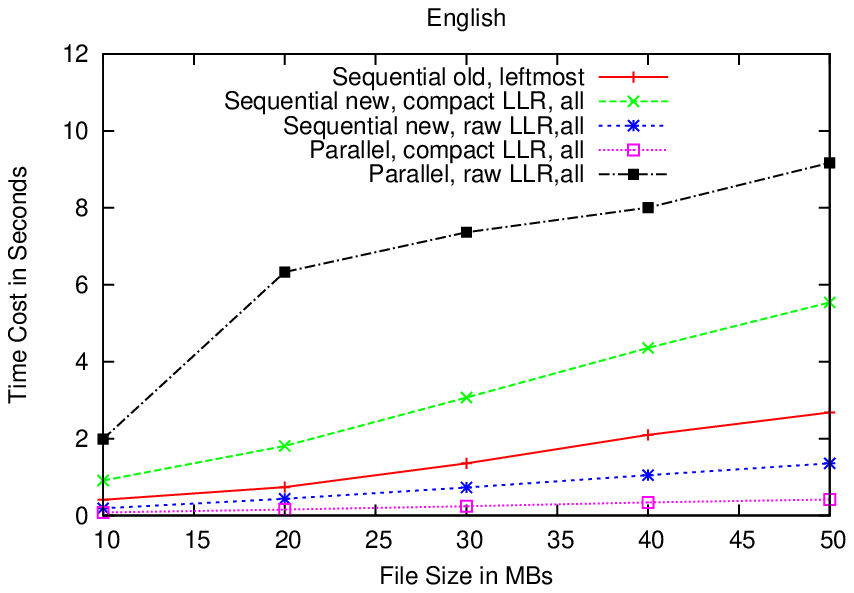} & 
\includegraphics[scale=.62]{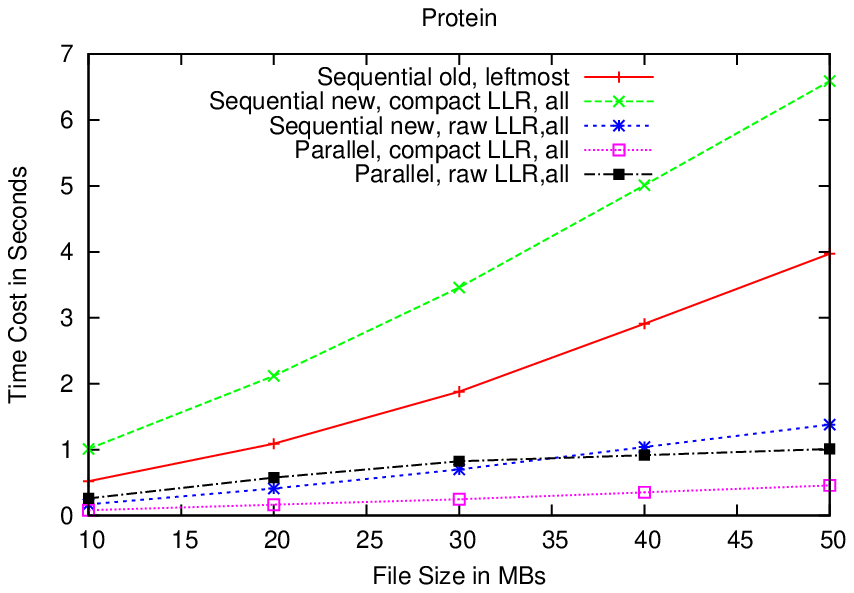} 
\end{tabular}
\end{center}
\caption{Time Cost vs Dataset Size. The three charts on the top and
  bottom show the experimental results
  on finding the leftmost and all repeats of every string position, respectively. }
\label{fig::time-size}
\end{figure}

\remove{

\begin{table}[t]
  \centering
  \begin{tabular}{|l|c|c|c|c|}
    \hline
    &Sequential & Sequential& Parallel & Parallel\\
    &No Compact &No Compact & Compact & Compact\\
    &Leftmost& All & Leftmost & All\\
    \hline
    DNA&$2.91$x&$2.91$x&$13.48$x&$\, \ 9.43$x\\
    \hline
    English&$1.97$x&$1.97$x&$\, \ 8.32$x&$\, \ 6.36$x\\
    \hline
    Protein&$3.44$x&$3.44$x&$14.62$x&$10.35$x\\   
    \hline
  \end{tabular}
  \caption{Speedup with 50MB Files}
  \label{tab:speedup}
\end{table}
\begin{table}[h!]
  \centering
  \begin{tabular}{|l|c|c|c|}
    \hline
    &Old (MBs) & Ours (MBs)& Space Saving \\
    \hline
    DNA & 792.77 & 650.39 & 17.96\% \\
    \hline
    English & 654.02&650.39&\,\ 0.56\% \\
    \hline
    Protein & 773.53&650.39&15.92\%\\ 
    \hline
  \end{tabular}
  \caption{RAM Usage Comparison for 50MB Files}
  \label{tab:speedup}
\end{table}

}

\begin{table}[t]
\parbox{.45\linewidth}{
\centering
  \begin{tabular}{|l|c|c|c|c|}
    \hline
    &Sequential & Sequential& Parallel & Parallel\\
    &No Compact &No Compact & Compact & Compact\\
    &Leftmost& All & Leftmost & All\\
    \hline
    DNA&$2.91$x&$2.91$x&$13.48$x&$\, \ 9.43$x\\
    \hline
    English&$1.97$x&$1.97$x&$\, \ 8.32$x&$\, \ 6.36$x\\
    \hline
    Protein&$3.44$x&$3.44$x&$14.62$x&$10.35$x\\   
    \hline
  \end{tabular}
  \caption{Speedup with 50MB Files}
  \label{tab:speedup}
}
\hfill
\parbox{.45\linewidth}{
\centering
\vspace*{7.5mm}
  \begin{tabular}{|l|c|c|c|}
    \hline
    &Old (MBs) & Ours (MBs)& Space Saving \\
    \hline
    DNA & 792.77 & 650.39 & 17.96\% \\
    \hline
    English & 654.02&650.39&\,\ 0.56\% \\
    \hline
    Protein & 773.53&650.39&15.92\%\\ 
    \hline
  \end{tabular}
  \caption{RAM Usage Comparison for 50MB Files}
  \label{tab:mem}
}
\end{table}

\subsection{Space}
\label{subsec:space}
Table~\ref{tab:mem} shows the peak memory usage of both the old and
our new algorithms for datasets of size $50$MBs.  The memory usage of
all of our algorithms is the same.  This is because the space usage by
the SA, Rank, and LCP array dominate the peak memory usage of all of
our algorithms. On the other hand, due to its 2-table system that
helps achieve the theoretical $O(n)$ time complexity, the old optimal
algorithm's space usage is relevant to the dataset type
 and is higher than ours.

\subsection{Scalability}
Although our algorithms have a superlinear time complexity in theory,
but they all scale well in practice as shown by
Figure~\ref{fig::time-size}.  As we increase the size of the test
data, we observe a consistent speedup. 
In addition, we \emph{did} conduct experiments on datasets of 100MB on the GPU device
by using a 2D grid of CUDA threads in order to create more than 100 million threads on the device.
When finding the leftmost LR for each string position, we observed the same speedups as shown in Figure~
\ref{fig::time-size}. 

 On the host CPU, the large
cache size dramatically reduces the total number of memory reads
during the linear walk in a raw LLR array and thus enables us to
eliminate the expensive binary search operations by using a raw LLR
array. On the GPU device, although all data is stored in the global
memory, a compact LLR array helps greatly reduce the total number of
global memory access; each thread linearly searches a smaller number
of LLR's. As shown in Table~\ref{tab:steps}, the average number of
walk steps in a compact LLR array is no more than six, which enables
the linear walk to be considered as a \emph{constant}-time operation.


\section{Conclusion and Future Work}
We proposed conceptually simple and thus easy-to-implement solutions
for longest repeat finding over a string. Our algorithm although is
not optimal in time theoretically, but runs faster than the old
optimal algorithm and uses less space. Further, our algorithm can find
all longest repeats of every string position, whereas the old optimal
solution can only find the leftmost one. Our algorithm can be
parallelized in shared-memory architecture and has been implemented on
GPU using the data parallelism to gain further speedup.

Our GPU solution is roughly 4.5 times quicker than our \emph{best}
sequential solution on the CPU, and up to 14.6 times quicker than the
old optimal solution on the CPU. Also, we improve the LLR compaction
performance by a factor of 40 on GPU. The multiprocessors in our
current GPU have a built-in L1 and L2 cache, which help coalesce some
global memory accesses. In the future, we will further optimize our
parallel solution by utilizing the GPU shared memory or texture memory
to further reduce global memory access.



\bibliographystyle{splncs03}
\bibliography{bibjsv,repeat,pm}

\newpage

\appendix

\section*{Appendix}

\remove{

\begin{table}[h]
\centering
\def\0{\phantom{0}}
{\footnotesize
\begin{tabular}{c|c|c|l}
\toprule
$i$ & $\lcp[i]$  & $\mathit{\sa}[i]$ & suffixes\\
\hline
\hline
\01 & 0 & 11\0  &{\tt i}\\ 
\02 & 1 & \08\0  & {\tt  ippi}\\ 
\03 & 1 & \05\0  & {\tt  issippi}\\ 
\04 & 4 & \02\0  & {\tt  ississippi}\\ 
\05 & 0 & \01\0  & {\tt  mississippi}\\
\06 & 0 & 10\0  & {\tt  pi}\\ 
\07 & 1 &  \09\0  & {\tt ppi}\\
\08 & 0 & \07\0  & {\tt sippi}\\
\09 & 2  & \04\0  & {\tt sissippi}\\
10 & 1  & \06\0  & {\tt ssippi}\\ 
11 & 3 & \03\0  & {\tt ssissippi}\\
12 & 0 & -- & --\\ \bottomrule
\hline
\end{tabular}
}
\caption{The suffix array and the lcp array of an example string $S={\tt mississippi}$.}
\label{tab:suflcp}
\end{table}

}


\begin{algorithm}[h!]
{\scriptsize
  \caption{Sequential finding of all $\lr_k$ for $k=1,2,\ldots,n$, using the raw LLR array.}
\label{algo:seq-1-ext}

\KwIn{The rank array and 
      the lcp array of the string $S$} 

\smallskip 

\tcc{Calculate $\llr_1,\llr_2,\ldots,\llr_n$.}

\lFor{$i=1,2,\ldots,n$}{
  $\llrr[i] \leftarrow \max\{\lcp[\rank[i]],\lcp[\rank[i]+1]\}$\tcp*{Length
    of $\llr_i$}

}

\smallskip

\tcc{Calculate $\lr_1,\lr_2,\ldots,\lr_n$.}
\For{$k=1,2,\ldots,n$}{

\tcc{Calculate the length of $\lr_k$.}
$length \leftarrow 0$ \tcp*{Length of $\lr_k$}

\For{$i = k$ down to $1$}{

  \If(\tcp*[f]{$\llr_i$ does not exist or does not cover $k$.}){$i+\llrr[i]-1<k$}{break\tcp*{Early stop}}
  \lElseIf{$\llrr[i] \geq length$}
   {$length\leftarrow \llrr[i]$\;}
}

\tcc{Calculate all $\lr_k$.}

\If{$length>0$}{
\For{$i = k$ down to $1$}{

  \If(\tcp*[f]{$\llr_i$ does not exist or does not cover $k$.}){$i+\llrr[i]-1<k$}{break\tcp*{Early stop}}
  \ElseIf{$\llrr[i] = length$}
   {print $\langle i,\llrr[i]\rangle$\tcp*{the start and ending positions
       of $\lr_k$}}
}

}

\lElse{
  print $\langle -1,0\rangle$\tcp*{$\lr_k$ does not exist. }  
}
}
}

\end{algorithm}



\begin{algorithm}[h!]
{\scriptsize
  \caption{Sequential finding of all $\lr_k$ for
    $k=1,\ldots,n$, using the LLRc array.}
\label{algo:seq-2-ext}

\KwIn{The rank array and 
      the lcp array of the string $S$} 

\smallskip 

\tcc{Calculate the compact LLR array.}
$j \leftarrow 1$; $prev\leftarrow 0$\;

\For{$i=1,2,\ldots,n$}{
  $L \leftarrow \max\{\lcp[\rank[i]],\lcp[\rank[i]+1]\}$\tcp*{Length
    of $\llr_i$}
  \lIf{$L> 0$ and $L \geq prev$}{
    $\llrc[j]\leftarrow \langle i, L\rangle$; $j\leftarrow j+1$\;
  }
  $prev\leftarrow L$\;
}
$size \leftarrow j-1$ \tcp*{Size of the $\llrc$ array. }

\smallskip

\tcc{Calculate $\lr_1,\lr_2,\ldots,\lr_n$.}
\For{$k=1,2,\ldots,n$}{
$start \leftarrow \textrm{BinarySearch}(\llrc,k)$\tcc{Return the
  smallest index of the $\llrc$ array element that covers position
  $k$, if such element exists; otherwise, return $-1$.}

\If{$start \neq -1$}{ 
\tcc{Find the length of $|\lr_k|$.}
$length \leftarrow 0$\;

  \For{$i = start \ldots size$}{
    
    \If(\tcp*[f]{$\llrc[i]$ does not cover
      $k$.}){$\llrc[i].start+\llrc[i].length-1<k$}
      {break\tcp*{Early stop}}
    \lElseIf{$\llrc[i].length > length$}
    {$length\leftarrow \llrc[i].length$\;}
  }
\tcc{Report all $\lr_k$.}
  \For{$i = start \ldots size$}{
    
    \If(\tcp*[f]{$\llrc[i]$ does not cover
      $k$.}){$\llrc[i].start+\llrc[i].length-1<k$}
      {break\tcp*{Early stop}}
    \ElseIf{$\llrc[i].length = length$}
    {print $\langle \llrc[i].start, \llrc[i].length\rangle$ 
          \tcp*{$\langle start, end\rangle$: start and ending
            pos of $\lr_k$.}}
  }
}
\lElse{print $\langle -1,0\rangle$\;}
}
}
\end{algorithm}


\end{document}